\documentclass[11pt]{article}
\usepackage{alec}

\DeclareMathOperator{\cost}{cost}
\DeclareMathOperator{\PoU}{PoU}

\theoremstyle{remark}
\newtheorem{CaseThm}{Case}[thm]

\begin{document}

\author{Yunzhe Bai \and Alec Sun}

\title{The price of uncertainty for social consensus}

% \keywords{Social networks, consensus, price of uncertainty, social welfare, game theory}

\maketitle

\begin{abstract}
    How hard is it to achieve consensus in a social network under uncertainty? In this paper we model this problem as a social graph of agents where each vertex is initially colored red or blue. The goal of the agents is to achieve consensus, which is when the colors of all agents align. Agents attempt to do this locally through steps in which an agent changes their color to the color of the majority of their neighbors. In real life, agents may not know exactly how many of their neighbors are red or blue, which introduces uncertainty into this process. Modeling uncertainty as perturbations of relative magnitude $1+\varepsilon$ to these color neighbor counts, we show that even small values of $\varepsilon$ greatly hinder the ability to achieve consensus in a social network. We prove theoretically tight upper and lower bounds on the \emph{price of uncertainty}, a metric defined in previous work by Balcan et al. to quantify the effect of uncertainty in network games.
\end{abstract}

\section{Introduction}

In the past decade, the rapid expansion of social media platforms has resulted in the formation of large and highly connected social networks. The evolution of these networks has been extensively studied in order to better understand how opinions held by agents are affected by their interactions with other people. In a social setting, individuals often modify their beliefs to more closely align with others that they frequently interact with \citep{social_influence, political_discourse}. This generally causes opinions held by agents in a network to converge over time as interacting agents attempt to reach a consensus. However, previous research has shown that additional factors, such as a sparse network structure or agents who refuse to change their opinions, can hinder consensus formation \citep{10.1287/moor.1120.0570, opinion_dynamics_bounded_confidence, social_influence}. The vast majority of social network models in the existing literature assume that agents have accurate information about the opinions held by others. However, in large networks where vertices generally have many neighbors and where interactions between agents are often brief, agents are often unable to make accurate judgments about the opinions of other agents, such as the number of their consenting or dissenting peers \citep{twitter_empirical}. 

Previous work modeling uncertainty in games has focused on the failure of agents to always pick the action with the optimal payoff. For example, in the quantal response model, agents sometimes deviate from optimal behavior, with strategies providing higher payoffs being chosen more often than those providing lower payoffs \citep{penna_noisy_best_response, BLUME1993387, logit_response_dynamics}. Alternatively, agents may occasionally fail to participate in the game, which can sometimes increase the social cost \citep{agent_failures, agent_resource_failures}.

In this paper, we instead focus on the effect of uncertainty in agents' \emph{observations} rather than their actions. We model a social network as a graph with an agent at each vertex. Individuals have a choice between two colors, red or blue, and the goal of each agent is to be the same color as the majority of their neighbors. This is called a \emph{consensus game} \citep{Convergence_in_Potential_Games, Fabrikant2004TheCO}. We consider a model of uncertainty first proposed by \citet{Balcan-2009}, where agents experience small, adversarial perturbations of relative magnitude $1+\eps$ to their observations of their neighbors. These perturbations capture the difficulty of ascertaining accurate information in large social networks, as the absolute margin of error increases proportionally to the number of neighbors an agent has. Alternatively, this model can also capture the effect of biased measurements, where agents may distort observations to agree with their beliefs (confirmation bias) or place undue emphasis on new information (novelty bias) \citep{confirmation_bias, novelty_bias}. 

The goal of this paper is to study how uncertainty affects behavior in consensus games as agents make potentially suboptimal updates to their colors. \citet{Balcan-2009} define a social cost function as the number of edges whose endpoints have different colors. They quantify the effect of uncertainty using a measure called the \emph{price of uncertainty} (PoU) of consensus games, defined as the maximum multiplicative factor by which the social cost can increase from any starting state. Even extremely small adversarial perturbations can snowball out of control and cause the game to move to states of much higher social cost \citep{Balcan-2011}. 

Our main results are tight asymptotic upper and lower bounds on the PoU for consensus games. Letting $n$ denote the number of agents participating in the consensus game, the previously best known asymptotic bounds due to \citet{Balcan-2011} were a lower bound of $\Omega(n^2\eps^3)$ and an upper bound of $O(n^2\eps)$. In this paper, we extend their work by proving that the PoU for consensus games is exactly $\Theta(n^2\eps^2)$. We believe that our result also holds irrespective of agent update ordering (i.e., the adversary can only control the magnitude of the perturbations), and this is a potential future direction for research.

% as long as $\eps$ does not decrease too quickly as $n$ grows. 

% Both bounds hold even for arbitrary orderings of player updates.

\section{Model}

We consider a consensus game in which each of $n$ agents is located at a distinct vertex of a simple undirected graph with $n$ vertices \citep{Convergence_in_Potential_Games}. In this paper, we will use the words \emph{agent} and \emph{vertex} interchangeably. Each agent can choose a color from $\{0,1\}$. Let us refer to two agents as \emph{neighbors} if and only if there exists an edge between them. An agent's cost is defined as the number of its neighbors that pick a different color; we call these \emph{bad} neighbors. We will also call edges whose endpoints have different colors \emph{bad} edges, and we will call neighbors or edges that are not bad \emph{good}. Finally, for any state $S$ of the game, we define a social cost function $\cost(S)$ as the sum of the costs of all agents in the game.

Starting from a ground state $S_0$ with social cost at least 1, let $S_t$ be a state achieved after $t$ steps under uncertain best-response dynamics. Here, the term \emph{uncertain} refers to the fact that an agent might miscount its cost by a multiplicative factor of at most $1+\eps$ for an uncertainty factor $0<\eps<1$ that can depend on $n$. Define an \emph{uncertain best response} as a best response that strictly increases the social cost as a result of uncertainty. In particular, the number of bad neighbors of the agent playing an uncertain best response must not exceed the number of good neighbors by more than a multiplicative factor of $1+\eps$. Also, by definition, an uncertain best response must involve a color switch. Note that we can assume $\eps = \Omega\bp{\fr{1}{n}}$, otherwise an agent cannot miscount its number of bad neighbors and hence there is no uncertainty. The \emph{price of uncertainty} (PoU) of consensus games is defined as
\[\PoU(\eps,\text{consensus})=\max\left(\frac{\cost(S_t)}{\cost(S_0)}\right),\]
where the maximum is taken over all initial states $S_0$ and all states $S_t$ reachable from $S_0$ in some number $t$ of steps \citep{Balcan-2009}.

This paper obtains a lower bound of $\Omega\bp{\eps^2n^2}$ on the PoU of consensus games for all $\eps=\til{\Omega}(n^{-1/2})$. Here $\til{\Omega}$ hides factors that are logarithmic in $n$. This is achieved using an intricate construction that amplifies the ``snowball'' effect of uncertainty described in \citet{Balcan-2011}. Next, by analyzing a clever construction, we prove that the PoU for consensus games is also bounded from above by $O\bp{\eps^2n^2}$ for all $\eps=\Omega(n^{-1/4})$. These bounds coincide, so we conclude that the asymptotically tight PoU for consensus games is $\Theta(\eps^2 n^2)$ for $\eps=\Omega(n^{-1/4})$, resolving an open problem of \citet{Balcan-2011}.

\subsection{Technical challenges} \label{challenges}

\paragraph{Lower bound.} In order to achieve a lower bound of $\Omega(\eps^3n^2)$, \citet{Balcan-2011} utilizes an \emph{initializer} gadget to set up initial colorings for an \emph{output} component. This component consists of successive layers of vertices, each of which has more vertices than the previous layer by a factor of $(1+\eps)$, and every vertex is connected to all vertices on adjacent layers. This output component is responsible for generating $\Omega(\eps n^2)$ bad edges, and the initializer gadget requires $\Theta(1/\eps^2)$ starting bad edges, so this achieves a lower bound of $\Omega(\eps^3n^2)$. We have found that it is very difficult to increase the final number of bad edges produced by such an output for any construction because the graph immediately becomes too connected for enough agents to play uncertain best responses. Rather, we present an alternative construction for an initializer gadget in \cref{proof-lower-bound} that only requires $\Theta(1/\eps)$ starting bad edges, which results in a construction for which the number of bad edges can increase by a factor of $O(\eps^2n^2)$.

\paragraph{Upper bound.} Generally, the goal of an adversary that wants to blow up the number of bad edges can be thought of as trying to maximize the number of additional bad edges that each new vertex produces in the uncertain best response process. Consider a time in which there are $m^2$ bad edges. The most ``compact'' way to store these bad edges using the least number of vertices is a bipartite graph $K_{m,m}$. If each of the $m$ vertices on one side of the bipartite graph plays an uncertain best-response, we can blow up the number of bad edges by a multiplicative factor of $(1+\eps)$. However, this requires $(1+\eps)m$ new vertices to ``absorb'' the new bad edges produced by the uncertain best-response. Hence, increasing the number of bad edges from $m^2$ to $(1+\eps)m^2$ requires roughly $(1+\eps)m$ new vertices.

Following this heuristic for a constant $\eps$ (i.e., not depending on $n$), the last stage will have a $K_{\Theta(\eps n), \Theta(\eps n)}$. Let us be optimistic in the sense that we assume that the original bad edges at each stage, which become good after vertices play an uncertain best response, can be turned bad again. Assuming most edges can eventually be turned bad again, as is the case in our lower bound construction, the total number of bad edges at the end is $$O\bp{\eps^2 n^2 + \fr{\eps^2 n^2}{1+\eps} + \cds} = O(\eps n^2).$$ Since there must be $\Omega(1/\eps)$ bad edges at the start, we estimate the upper bound on the price of uncertainty for consensus games to be $O(\eps^2 n^2)$. 

It appears natural to track the number of bad edges at every step and try to establish an upper bound on its growth rate as new vertices play uncertain best responses. The hope is that we can prove an upper bound on the growth rate in line with our estimate of $O(\eps^2n^2)$. Unfortunately, this approach fails, as illustrated by the following example.

\begin{example} \label[example]{double}
Consider a red vertex $V$ whose only neighbors are $m$ blue vertices and $m+1$ red vertices for $m = \Omega(1/\eps)$. The $m$ blue vertices are all connected to another blue vertex $B$, and the $m+1$ red vertices are all connected to another red vertex $R$. Then, we can have $V$ play $2m+1$ uncertain best responses in order to roughly triple the number of bad edges while only operating on these $O(m)$ vertices.
\begin{itemize}
    \item Let $V$ play an uncertain best response that switches its color to blue, resulting in $m+1$ bad edges.
    
    \item Take the endpoint of one such bad edge adjacent to $R$ and switch its color so that there are now $m$ bad edges incident to $V$ and 1 bad edge incident to $R$.
    
    \item Switch the color of $V$ again, resulting in $m+1$ bad edges. Take the endpoint of one such bad edge adjacent to $B$ and switch its color so that there are now $m$ bad edges incident to $V$, 1 bad edge incident to $B$, and 1 bad edge incident to $R$.
    
    \item Repeat the above, switching at each step the color of a neighbor of $V$ whose color has not been switched before, until all neighbors of $V$ have had their colors switched once. There are now $m$ bad edges incident to $B$ and $m+1$ bad edges incident to $R$.
\end{itemize}
\end{example}
Note that \cref{double} can be generalized to operate on a $K_{m,2m+1}$ where each vertex in the $m$-size part is blue, $m$ of the vertices in the $(2m+1)$-size part are red, and $m+1$ of the vertices in the $(2m+1)$-size part are blue to also roughly triple the number of bad edges. The bad edge multiplicative growth rate of this process operating on $\Theta(m)$ vertices is 3, whereas in our initial estimate it was $1+\eps$. Hence, using this analysis, it is actually possible for there to be $\Omega(n^2)$ bad edges at the last stage if operating on $\Theta(m)$ vertices can triple $\Theta(m^2)$ bad edges at each step. We thus cannot prove an upper bound of $O(\eps n^2)$ on the number of bad edges at the last stage using this strategy.

Instead, one idea is that we should focus not on the total number of bad edges at every step but rather consider the step at which each vertex \emph{first} plays an uncertain best response. This imposes a sequential order on the vertices from the first vertex to play an uncertain best response to the last. Recall that in order for a vertex with high degree to be able to play an uncertain best response, it must be the case that a large number of edges adjacent to this vertex must be bad. If we are to end up with $\Omega(\eps n^2)$ bad edges at the end, it must be the case that the later vertices in this sequence must have high degree on average. This requires a large number of existing bad edges to be generated from the vertices before the final vertices in the sequence, which in turn requires bad edges from the vertices before them, and so on.

Our proof of the $O(\eps^2 n^2)$ upper bound analyzes the degrees of vertices in this sequence and shows that the sequence of degrees cannot increase ``too fast'' if vertices can only play $(1+\eps)$-uncertain best responses, where we realize the notion of ``too fast'' via a clever construction.

\section{Proof of lower bound} \label{proof-lower-bound}

We first improve the lower bound on the PoU for consensus games. Furthermore, we do so in the expanded domain $\eps = \til{\Omega}(n^{-1/2})$. Here, $\til{\Omega}$ hides factors that are logarithmic in $n$.

\begin{theorem} \label{thm:lower-bound}
    The PoU for consensus games is $\Omega(\eps^2 n^2)$ for $\eps = \til{\Omega}(n^{-1/2})$.
\end{theorem}

\citet{Balcan-2011} used a carefully tuned gadget to create a snowball-like effect of uncertainty. To prove \cref{thm:lower-bound}, we provide a more complicated construction that amplifies the snowball effect even further. One general idea is to construct a graph as a sequence of levels, where each level contains slightly more vertices than the previous level. At the beginning, all of the vertices in the first level should have enough bad edges to play uncertain best responses, which increases the number of bad edges by a multiplicative factor of $1+\eps$. The fact that the vertices in the first level switched color implies that the vertices in the second level now have enough bad neighbors to play an uncertain best response, and this effect snowballs across the levels.

We will make the above idea rigorous. We define a \emph{boosting operation} between layers of the following form. Let $\ell>1$ and suppose that level $\ell-1$ has $\fr{1}{1+\fr{\eps}{2}}\cd m$ vertices for some $m$, level $\ell$ has $m$ vertices, and level $\ell+1$ has $\bp{1+\fr{\eps}{2}}\cd m$ vertices. Suppose that every vertex in level $\ell$ is connected to every vertex in level $\ell-1$ and level $\ell+1$, so the degree of every vertex in level $\ell$ is $d_\ell = \bp{\fr{1}{1+\fr{\eps}{2}} + (1+\fr{\eps}{2})} \cd m$. If the vertices in levels $\ell$ and $\ell+1$ are all blue and the vertices in level $\ell-1$ are all red, then the edges between levels $\ell-1$ and $\ell$ are all bad. Since $(1+\eps) \cd \fr{1}{1+\fr{\eps}{2}}\ge \fr{d_\ell}{2}$ for all $\eps \le 1$, then all vertices in level $\ell$ are allowed to play uncertain best responses in sequence. This turns level $\ell$ all red. Note now that the number of bad edges between levels $\ell$ and $\ell+1$ is $\bp{1+\fr{\eps}{2}}\cd m^2$, which is a $1+O(\eps)$ multiplicative increase from the number of original bad edges $\fr{1}{1+\fr{\eps}{2}}\cd m^2$ from before this boosting operation.

We now present the construction. In order to implement the idea above, we must first find some way to create the first two layers of size $\fr{1}{\eps}$ because a vertex must have degree at least $\fr{1}{\eps}$ in order to play an uncertain best response. The \emph{initializer} layers will have numbers of vertices 
\[\fr{1}{\eps}, 1, \fr{1}{\eps} + 1, 1, \lds, \fr{2}{\eps} -1, 1, \fr{2}{\eps}, 2, \fr{1}{\eps}, 2, \fr{1}{\eps} + 1, \lds, \fr{2}{\eps} -1, 2, \fr{2}{\eps}, 4, \fr{1}{\eps}, \lds, \fr{1}{\eps}, \fr{1}{\eps},\] 
where the numbers of vertices in the odd-indexed layers cycle through the consecutive integers from $\fr{1}{\eps}$ to $\fr{2}{\eps}$, and the numbers of vertices in even-indexed layers increase by a factor of 2 whenever the odd-indexed layers' cycle repeats up until the number reaches $\fr{1}{\eps}$. Each vertex is connected to all vertices in the layer immediately after it except for vertices in odd-indexed layers with $\fr{2}{\eps}$ vertices, which are connected to the next layer to form a regular bipartite graph with the next layer where the vertices in the next layer are each connected to $\fr{1}{\eps}$ vertices in the current layer. Note that the total number of
initializer layers is $\fr{1}{\eps} \cd \log_2 \fr{1}{\eps}$, and each layer has at most $\fr{2}{\eps}$ vertices, so the total number of vertices in the initializer layers is $O\bp{\fr{1}{\eps^2} \log \fr{1}{\eps}}$. The initializer sequence has been constructed in such a way such that by starting from a state in which the first layer is red and all other layers are blue, the vertices in each layer can play uncertain best responses to change their color to red in sequence so that the last two layers form a $K_{\fr{1}{\eps}, \fr{1}{\eps}}$ where the left side is red and the right half is blue. Finally, we construct $\fr{1}{\eps}$ blue vertices in a parallel layer to the red left side of the $K_{\fr{1}{\eps}, \fr{1}{\eps}}$ and connect each of them to every vertex in the right half of the $K_{\fr{1}{\eps}, \fr{1}{\eps}}$, so that each vertex in the right half now has $\fr{1}{\eps}$ red neighbors and $\fr{1}{\eps}$ blue neighbors on the left. Note that no new bad edges are added. We will now leave all layers before the last layer frozen: no vertex in a layer before the last layer will ever change color again.

After the last layer in the initializer, construct a \emph{secondary} set of layers with numbers of vertices
\[\fr{1}{\eps}, 1, \fr{1}{\eps} + 1, 1, \lds, \fr{2}{\eps} -1, 1, \fr{2}{\eps}, 2, \fr{1}{\eps}, 2, \fr{1}{\eps} + 1, \lds, \fr{2}{\eps} -1, 2, \fr{2}{\eps}, 4, \fr{1}{\eps}, \lds, \fr{2}{\eps}, \fr{2}{\eps}.\]
Note that this sequence is almost identical to that of the initializer, except we keep going until a $K_{\fr{2}{\eps}, \fr{2}{\eps}}$ is formed by the last two layers. The first layer in the secondary set is the same as the last layer of the initializer, each vertex of which is connected with $\fr{1}{\eps}$ frozen blue vertices and $\fr{1}{\eps}$ frozen red vertices. We now have the following observation:

\begin{proposition} \label[proposition]{freely-change}
    No matter what happens in future layers, each of the $\fr{1}{\eps}$ vertices in the last layer of the initializer can freely change colors.
\end{proposition}

\begin{proof}
    Each vertex in this layer has total degree $\fr{2}{\eps}+1$ and will always have at least $\fr{1}{\eps}$ frozen neighbors of every color, which is enough to play an uncertain best response no matter what the current color is.
\end{proof}

Our careful construction of the initializer and \cref{freely-change} are the key observations needed to improve the bound in \citet{Balcan-2011}. We will now finish the construction. Note that the $K_{\fr{2}{\eps}, \fr{2}{\eps}}$ produced at the end of the secondary set is a starting point from which we can perform the \emph{boosting operations} as defined above. Here we add a layer of $\fr{2}{\eps} + 1$ vertices to form an initial $K_{\fr{2}{\eps}, \fr{2}{\eps} + 1}$, and then initialize $m$ to be $\fr{2}{\eps} + 1$. Starting from this layer, additional layers will increase in size by a constant multiplicative factor of $1 + \fr{\eps}{2}$ until we reach $c \cd \eps n$ vertices in the last layer for some small constant $c$ to be set later. The total number of vertices used in the boosting operations is
\[\Theta\bp{\fr{c \cd \eps n}{1 - \fr{1}{1 + \fr{\eps}{2}}}}=\Theta\bp{\fr{c \cd \eps n}{\frac{\eps}{2+\eps}}}=\Theta(cn).\] 
Combined with the number of vertices $O\bp{\fr{1}{\eps^2} \log \fr{1}{\eps}}$ used in the initializer as well as the secondary set, we see that for $\eps = \til{\Omega}(n^{-1/2})$ and a constant $c$ independent of $\eps$ and $n$ that the total number of vertices used is at most $n$, which means our construction has the right number of vertices.

How many bad edges does our construction yield? Starting from the first layer in the secondary set, we can turn each successive layer red in sequence so that the last two layers form a $K_{c\cd \eps n, c\cd \eps n}$ where the left side is red and the right half is blue. However, this produces only $c^2\cd \eps^2 n^2$ bad edges, which, recalling that we started with $\fr{1}{\eps}$ bad edges, matches the bound $\Omega(\eps^3 n^2)$ on the price of uncertainty in \citet{Balcan-2011}. It turns out we can produce many more bad edges in our construction by doing the following. Note that the last layer in our construction is blue, and all layers between the frozen layers defined above and the last layer are now red. We now turn each vertex in the layer after the frozen layers blue, which is a valid uncertain best response by \cref{freely-change}. We can now turn each successive layer blue in sequence until the last three layers are blue, red, and blue, respectively. By repeatedly switching the colors of the vertices in the layer after the frozen layers and propagating uncertain best responses, we can reach a state in which the boosting operation layers alternate between blue and red. Every edge between boosting operation layers is now bad, which makes
\[\Theta\bp{\fr{c^2 \cd \eps^2 n^2}{1 - \bp{\fr{1}{1 + \fr{\eps}{2}}}^2}}=\Theta\bp{\fr{c^2 \cd \eps^2 n^2}{\fr{\eps(4+\eps)}{(2+\eps)^2}}}=\Theta(\eps n^2)\] 
bad edges at the ending state. Recalling that we began with $\fr{1}{\eps}$ bad edges, we have established that the price of uncertainty for consensus games is $\Omega(\eps^2 n^2)$ in the domain $\eps = \til{\Omega}(n^{-1/2})$.

\section{Proof of upper bound}

Our second main result is establishing the upper bound on the price of uncertainty.

\begin{theorem} \label{thm:upper-bound}
The price of uncertainty in consensus games satisfies
\[
\PoU(\eps,\text{consensus}) = O(\eps^2 n^2)
\quad \text{for } \eps = \Omega(n^{-1/4}).
\]
\end{theorem}

As discussed in \cref{challenges}, we focus only on the steps at which each vertex plays its first uncertain best response. In addition, we also carefully construct and analyze a construction that restricts the degrees of new vertices and allows us to bound the number of edges in the graph. Let $G$ be a graph for a consensus game, and let $V$ be the set of vertices of $G$. It will be useful to first establish a bound on the starting number of bad edges.

\begin{claim}\label[claim]{start-bad-edges}
    Suppose that it is possible to increase the number of bad edges in $G$ by a factor of $\Omega\bp{\eps^2n^2}$. Then, the starting number of bad edges must be $O\bp{\eps^{-2}}$.
\end{claim}

\begin{proof}
    There can be at most $O\bp{n^2}$ bad edges in a graph with $n$ vertices, so the starting number of bad edges must be
    \[\frac{O\bp{n^2}}{\Omega\bp{\eps^2n^2}}=O\bp{\eps^{-2}}.\]
\end{proof}

We begin by partitioning $V$ into three disjoint subsets:
\begin{itemize}
    \item Let $P$ be the set of vertices in $V$ that begin without any bad edges and play at least one best response.
    \item Let $Q$ be the set of vertices in $V$ that begin without any bad edges and never play a best response.
    \item Let $S_0$ be the set of all remaining vertices in $V$, i.e. those that begin incident to at least one bad edge.
\end{itemize}
Let us number the vertices in $P$ in order of their first best responses from first to last as $p_1,p_2,\dots,p_m$ for $m=|P|$. Then, for all integers $1\leq k\leq m$, define $S_k=S_0\cup\{p_1,p_2,\dots,p_k\}$. Notice that the first vertex $p_1$ can only change its color after several vertices in $S_0$ change colors, because $p_1$ must begin with only good edges. Now that we have set up our notation, we will prove some preliminary statements regarding the behavior of the consensus game.

\begin{proposition}\label[proposition]{contained-bad-edges}
    Let $k$ be an integer such that $0\leq k\leq m$. Before $p_{k+1}$ plays its first best response, every bad edge is incident to at least one vertex in $S_k$. Furthermore, every bad edge is always incident to at least one vertex in $S_m$.
\end{proposition}

\begin{proof}
    Before any best responses are played, all bad edges have endpoints in $S_0$ by definition. Similarly, at this point, edges between two vertices in $V\setminus S_0$ must all be good.
    
    Let $0\leq j\leq m-1$ be an integer. Suppose that \cref{contained-bad-edges} holds when $k=j$. Then, any edge between two vertices in $V\setminus S_j$ is good. By construction, the first vertex outside of $S_j$ to play a best response is $p_{j+1}$. When this happens, the only edges to be affected will be those incident to $p_{j+1}$. However, since $p_{j+1}\in S_{j+1}$, every bad edge must now be incident to at least one vertex in $S_{j+1}$. Finally, as long as no vertices outside of $S_{j+1}$ change colors, every bad edge will remain incident to at least one vertex in $S_j$. By strong induction, \cref{contained-bad-edges} must hold for all integers $k$ between $0$ and $m-1$ inclusive.

    Finally, by construction, all vertices in $V\setminus S_m=Q$ never play a best response. This means that any edges between 2 of these vertices will remain good, so \cref{contained-bad-edges} holds for $k=m$ as well.
\end{proof}

\begin{lemma}\label[lemma]{m-equals-n}
    Suppose that $\eps=\omega(n^{-1/2})$. Then, the factor by which the number of bad edges in $G$ can be increased is maximized when $m=\Theta(n)$.
\end{lemma}

\begin{proof}
    Clearly, $m=|P|<|V|=n$, so $m=O(n)$ and the upper bound holds. We will now prove the lower bound.
    
    By construction, $V$ consists of the two disjoint subsets $S_m$ and $Q$. By \cref{contained-bad-edges}, any vertices in $Q$ that do not share an edge with a vertex in $S_m$ can never be incident to a bad edge. This means that we can delete these vertices along with all of their incident edges, which will never decrease the factor by which the number of bad edges in $G$ can be increased. We will take this a step further by deleting all vertices in $Q$ that do not share an edge with a vertex in $P$. By similar logic, this cannot decrease the PoU because all vertices in $V\setminus Q$ can still have at most as many good edges as before. Thus, the same sequence of uncertain best responses can be played to produce the same number of final bad edges even though we now have fewer vertices in $V$.

    After all of our deletions, we are now left with a new set of vertices $Q'\subset Q$ where every edge incident to a vertex in $Q'$ is also incident to a vertex in $S_m\setminus S_0=P$. Then, there exists a vertex $v\in P$ with the most edges connecting it to a vertex in $Q'$. Let the number of such edges be $w$. We then construct a set $R$ of vertices $\{r_1,r_2,\dots,r_w\}$, and we can add new edges such that each vertex in $P$ connected by $j$ edges to vertices in $Q'$ is connected to $r_1,r_2,\dots,r_j$. Finally, we delete all vertices in $Q'$ and all of their incident edges. The factor by which the number bad edges in $G$ can be increased doesn’t change because all vertices that change colors still have the same number of good edges. Furthermore, $|R|=w\leq|Q'|$, so the PoU cannot decrease. Finally, since $v$ always has at least $w=|R|$ good edges, it must have at least $|R|(1+\eps)^{-1}=\Theta(|R|)$ bad edges when it plays its first uncertain best response. However, by construction, $v$ cannot have more than $|S_m|$ bad edges, so $|R|=O(|S_m|)$. Since the entire graph now only consists of vertices in $R$ and $S_m$, this shows that $n=\Theta(|S_m|)$.
    
    Next, recall that $|S_m|=|S_0|+m$. By \cref{start-bad-edges}, we must have $O\bp{\eps^{-2}}$ edges that start as bad, and the number of vertices that begin with a bad neighbor is at most twice this quantity. So, $|S_0|=O\bp{\eps^{-2}}$ as well, and we thus have
    \begin{align*}
        \Theta(n)&=|S_m|=|S_0|+m=O\bp{\eps^{-2}}+m.
    \end{align*}
    Since $\eps=\omega\bp{n^{-1/2}}$, this becomes
    \[\Theta(n)=o\bp{n}+m\implies\Theta(n)=m.\]
\end{proof}

With our construction, \cref{contained-bad-edges} allows us to contain the bad edges in the graph at any step inside one of $S_0,S_1,\dots,S_m$. We will now present a combinatorial abstraction that will allow us to focus on the effect of the first best response played by each vertex in $P$.

For every integer $k$ satisfying $0\leq k\leq m$, we can use $S_k$ to construct a multiset $A_k$ to encode information about certain bad edges in $S_k$. Specifically, each element of $A_k$ will describe the number of edges between a particular vertex in $V\setminus S_k$ and any vertex in $S_k$, and each vertex in $V\setminus S_k$ with at least one connection to a vertex in $S_k$ will be encoded exactly once in such a manner. 

\begin{example}\label[example]{multiset-example}
    Suppose that there existed four edges between a vertex in $S_0$ and some vertex $x_1\not\in S_0$, and three edges between a vertex in $S_0$ and some different vertex $x_2\not\in S_0$. Also, suppose that there existed no other edges between a vertex in $S_0$ and a vertex in $V\setminus S_0$. Then we would have
    \(A_0=\{4,3\}.\)
\end{example}

\begin{figure}
      \centering
      \begin{tikzpicture}
        \draw (0,0) ellipse (0.6cm and 1.2cm) node[label={[yshift=1.2cm,black]above:$S_0$}] {};
        \filldraw[fill=black] (2,0.8) circle (0.1cm) node[label={[yshift=0.1cm,black]above:$x_1$}] {};
        \filldraw[fill=black] (2,-0.8) circle (0.1cm) node[label={[yshift=0.1cm,black]above:$x_2$}] {};
        \draw (0,0.9)--(2,0.8);
        \draw (0,0.5)--(2,0.8);
        \draw (0,0.2)--(2,0.8);
        \draw (0,0.2)--(2,-0.8);
        \draw (0,-0.2)--(2,0.8);
        \draw (0,-0.5)--(2,-0.8);
        \draw (0,-0.9)--(2,-0.8);
    \end{tikzpicture}
      \caption{In \cref{multiset-example}, four edges connect vertices in $S_0$ to $x_1$, three edges connect vertices in $S_0$ to $x_2$, and two of the edges are incident to the same vertex in $S_0$.}
    \end{figure}

Consider the state of the graph right before $p_k$ plays its first best response. By \cref{contained-bad-edges}, all edges between $p_k$ and a vertex in $V\setminus S_{k-1}$ must be good. If we let $z$ be the number of edges incident to $p_k$ and a vertex in $S_{k-1}$, then the degree of $p_k$ cannot be greater than $z+\lfloor(1+\eps)z\rfloor$. By this logic, going from $A_{k-1}$ to $A_k$ consists of deleting the element of $A_{k-1}$ corresponding to the edges between a vertex in $S_{k-1}$ and $p_k$, which is $z$. Then, we will select $\lfloor(1+\eps)z\rfloor$ elements of $A_{k-1}$ and increment them by $1$ (where we temporarily pad $A_{k-1}$ with zeroes to allow for the creation of new elements equal to 1). 

\begin{example}
    Suppose that $A_0=\{3,2,2\}$, $\eps=1/3$, and we wish to play a move on $z=3$. Then, $\lfloor(1+\eps)z\rfloor=4$, and the possible multisets for $A_1$ are:
\[\{3,3,1,1\},\quad\{3,2,1,1,1\},\;\;\text{and}\;\;\{2,2,1,1,1,1\}.\]
\end{example}

The new multiset $A_k$ represents the edges between $p_k$ and a vertex in $V\setminus S_k$. Since any graph and sequence of best responses can be used to construct the multisets $(A_k)_{0\leq k\leq m}$ in such a manner, it suffices to study the behavior of these moves.

\begin{definition}
    A \emph{regular move} on a multiset $A_k$ with a positive integer $z\in A_k$ consists of the following steps:
    \begin{itemize}
        \item Delete one instance of $z$ from $A_k$.
        \item Pad $A_k$ with an sufficiently large (at least $\lfloor z(1+\eps)\rfloor$) number of zeroes.
        \item Create $A_{k+1}$ from $A_k$ by selecting $\lfloor z(1+\epsilon)\rfloor$ elements of $A_k$ and adding 1 to each of them.
        \item Remove all remaining zeroes from $A_{k+1}$.
    \end{itemize}
\end{definition}

As a shorthand notation, let $\sum A_k$ denote the sum of the elements (with their multiplicities) of $A_k$, and let $\sum A_k^2$ denote the sum of squares of the elements (with their multiplicities) of $A_k$. The former will be a useful measure to bound because of \cref{sum-to-finish}, and the latter will assist us in obtaining such a bound.

\begin{proposition}\label[proposition]{sum-E0}
    $\sum A_0=O\bp{n|S_0|}$.
\end{proposition}

\begin{proof}
    Since there are only $n$ vertices in the graph, every vertex in $S_0$ can have at most $n$ edges. Recall that $\sum A_0$ counts the number of edges between a vertex in $S_0$ and a vertex in $V\setminus S_0$. \cref{sum-E0} follows.
\end{proof}

\begin{lemma}\label[lemma]{sum-to-finish}
    If $\sum A_m=O\bp{\eps^3n^2|S_0|}$, then the final number of bad edges in the graph is $O\bp{\eps^2n^2|S_0|}$ if $\eps=\omega(n^{-1/3})$.
\end{lemma}

\begin{proof}
    Let $T_k$ denote the number of edges between two vertices in $S_k$. For every regular move that we play on $A_{k-1}$ with an integer $z\in A_{k-1}$, the sum $\sum A_k$ grows from $\sum A_{k-1}$ by $\lfloor z(1+\eps)\rfloor-z$, and the number of edges between two vertices in $S_k$ grows from $S_{k-1}$ by $z$. This is because, by definition, $z$ denotes the number of edges that connect $S_k$ to the new vertex, and it is precisely these edges that are counted by $T_k$ and not $T_{k-1}$. Thus, we have
    \[\lfloor\eps z\rfloor=\lfloor z(1+\eps)\rfloor-z=\sum A_k-\sum A_{k-1},\]
    as well as
    \[\eps\bp{T_k-T_{k-1}}-1=\eps z-1<\lfloor\eps z\rfloor,\]
    which by substitution implies the inequality
    \[\eps\bp{T_k-T_{k-1}}-1<\sum A_k-\sum A_{k-1}.\]
    By summing both sides from $k=1$ to $m$, we get
    \begin{align*}
        \eps\left(T_m-T_0\right)-m&\leq\sum A_m-\sum A_0.
    \end{align*}
    Since $\sum A_m=O\bp{\eps^3n^2|S_0|}$,
    \begin{align*}
        \eps\left(T_m-T_0\right)-m&\leq O\bp{\eps^3n^2|S_0|}-\sum A_0.
    \end{align*}
    By \cref{m-equals-n}, this becomes
    \begin{align*}
        \eps\left(T_m-T_0\right)-\Theta(n)&\leq O\bp{\eps^3n^2|S_0|}-\sum A_0 \\
        \eps T_m&\leq O\bp{\eps^3n^2|S_0|}-\sum A_0+\Theta(n)+\eps T_0 \\
        T_m&\leq O\bp{\eps^2n^2|S_0|}-\eps^{-1}\sum A_0+\Theta\bp{\eps^{-1}n}+T_0.
    \end{align*}
    By \cref{sum-E0}, $\eps^{-1}\sum A_0$ must be $O(\eps^{-1}n|S_0|)$. Since we assumed that $\eps=\omega(n^{-1/3})$, this is dominated by the $O(\eps^2n^2|S_0|)$ term.
    Also, since we must have $\Omega(\eps^{-1})$ starting bad edges, $S_0$ is $\Omega(\eps^{-1})$ as well. This means that $O(\eps n^2)$ must also be $O(\eps^2n^2|S_0|)$, and then $\eps=\omega(n^{-1/3})$ implies that $O(\eps n^2)$ will dominate $\Theta(\eps^{-1}n)$. We are left with the following terms:
    \begin{align*}
        T_m&\leq O\bp{\eps^2n^2|S_0|}+T_0.
    \end{align*}

    We will now show that $T_0$ cannot grow faster than $O(\eps^2n^2|S_0|)$. By construction, every vertex in $S_0$ must begin incident to at least one bad edge, so the starting number of bad edges is at least $T_0/2=\Theta(T_0)$. By \cref{start-bad-edges}, we start with $O(\eps^{-2})$ bad edges, so $T_0=O(\eps^{-2})$. This is also $O(\eps n^2)$ because $\eps=\omega(n^{-1/3})$. Then, $O(\eps n^2)$ is $O(\eps^2n^2|S_0|)$ because $|S_0|=\Omega(\eps^{-1})$. Our equation finally becomes
    \begin{align*}
        T_m&=O\bp{\eps^2n^2|S_0|}.
    \end{align*}
    
    Any vertex in $V$ is either in $S_m$, or $Q$. This means that the only bad edges that are not counted by $T_m$ must be between $S_m$ and a vertex in $Q$, which are entirely described by the elements in $A_m$. Thus, the number of bad edges between a vertex in $S_m$ and a vertex in $Q$ is bounded by $\sum A_m$, which is $O(\eps^3n^2|S_0|)$. The total number of bad edges is then bounded by 
    \[T_m+\sum A_m=O\bp{\eps^2n^2|S_0|}+O\bp{\eps^3n^2|S_0|}=O\bp{\eps^2n^2|S_0|}.\]
\end{proof}

Intuitively, playing a regular move on $\sum A_{k-1}$ with a large $z$ would result in a big increase in from $\sum A_{k-1}$ to $\sum A_k$. However, playing such a move might also create many small elements in $A_k$. Having the sum $\sum A_k$ spread out across many small elements will make increasing $\sum A_k$ in future moves more difficult. Recall that $\sum A_k^2$ denotes the sum of squares of the not necessarily distinct elements of $A_k$. We consider $\sum A_k^2$ as a way to measure how "spread out" the multiset $A_k$ is, and this particular quantity is preferred over other measures because the change from $\sum A_{k-1}^2$ to $\sum A_{k}^2$ can be related to $\sum A_{k-1}$. 

Our plan is to bound the behavior of $\sum A_k$ and $\sum A_k^2$ as $k$ increases. The multisets $A_0$ and $A_m$ will be subject to the following inequalities:
\begin{align}
    \sum A_m^2&\geq0 \label{Em^2-greater-than-Em} \\
    \sum A_0^2&\leq\left(\sum A_0\right)^2 \label{E02-less-than-(E0)^2}.
\end{align}
\cref{Em^2-greater-than-Em} is true because the elements of $A_m$ must be nonnegative integers, and \cref{E02-less-than-(E0)^2} follows from expanding both sides. We will only consider these inequalities (and not those with $m$ and $0$ swapped) because we only need to prevent \emph{too many} regular moves from being played with large $z$, which will naturally cause $\sum A_k^2$ to become small as $k$ increases.

We will begin by investigating the behavior of the sum of squares. Let the current multiset be $A_{k-1}$. Then, pick some element $z\in A_{k-1}$, and define $\zeta=\lfloor z(1+\eps)\rfloor$. We can play a regular move on $A_{k-1}$ with $z$, which will involve incrementing the $\zeta$ elements $a_1,a_2,\dots,a_\zeta\in A_{k-1}$ to $a_1+1,a_2+1,\dots,a_\zeta+1\in A_k$. It is easy to see that the effect of such a move on the increase from $\sum A_{k-1}$ to $\sum A_k$ is
\begin{align}
    \sum A_k-\sum A_{k-1}&=-z+\sum_{i=1}^\zeta\big((a_i+1)-a_i\big) \nonumber \\
    \sum A_k-\sum A_{k-1}&=-z+\zeta. \label{new-sum-relation}
\end{align}
Next, we must have
\begin{align}
    \sum A_k^2-\sum A_{k-1}^2&=-z^2+\sum_{i=1}^\zeta\left((a_i+1)^2-a_i^2\right) \nonumber \\
    \sum A_k^2-\sum A_{k-1}^2&=-z^2+2\sum_{i=1}^\zeta a_i+\zeta. \label{new-sumq-relation}
\end{align}
\cref{new-sum-relation} and \cref{new-sumq-relation} together describe a single regular move. However, we can establish bounds on these equations that are easier to analyze. Since $\eps<1$, we have $\zeta=\lfloor z(1+\eps)\rfloor\leq2z$, so we can rewrite \cref{new-sumq-relation} as
\begin{align}
    \sum A_{k}^2-\sum A_{k-1}^2&\leq2\sum_{i=1}^\zeta a_i-z^2+2z. \nonumber
\end{align}
Next, since $a_1,a_2,\dots,a_\alpha$ are all elements of $A_{k-1}$, their sum must be bounded by $\sum A_{k-1}$. By plugging this into our expression for the change from $\sum A_{k-1}^2$ to $\sum A_k^2$, we get
\begin{align}
    \sum A_{k}^2-\sum A_{k-1}^2&\leq2\sum A_{k-1}-z^2+2z. \label{sum-squares-move}
\end{align}

\begin{proposition}\label[proposition]{optimal-bound-sum-squares}
    Taking the upper bound in \cref{sum-squares-move} for the increase from $\sum A_{k-1}^2$ to $\sum A_k^2$ never decreases $\sum A_m$.
\end{proposition}

\begin{proof}
    Notice that the changes described in the right hand sides of \cref{new-sum-relation}, \cref{new-sumq-relation}, and \cref{sum-squares-move} are all independent of the current value of $\sum A_{k-1}^2$. Furthermore, the only constraint on the final sum of squares $\sum A_m^2$ is a lower bound provided by \cref{Em^2-greater-than-Em}. This means that always taking our upper bound on $\sum A_k$ from \cref{sum-squares-move} will never decrease $\sum A_m$.
\end{proof}

Now equipped with \cref{optimal-bound-sum-squares}, we actually no longer need that $z\in A_{k-1}$. In fact, we no longer even need to consider $A_{k-1}$ as a multiset anymore. For the rest of the proof, we will simply assume that $z$ is a positive real number. The next step is to find an upper bound on the change from $\sum A_{k-1}$ to $\sum A_k$, which is
\begin{align}
    \sum A_k-\sum A_{k-1}&=\zeta-z \nonumber \\
    \sum A_k-\sum A_{k-1}&=\lfloor z(1+\eps)\rfloor-z \nonumber \\
    \sum A_k-\sum A_{k-1}&\leq z(1+\eps)-z+1 \nonumber \\
    \sum A_k-\sum A_{k-1}&\leq\eps z+1. \label{sum-move}
\end{align}

\begin{proposition}\label[proposition]{optimal-bound-sum}
    Taking the upper bound in \cref{sum-move} for the increase from $\sum A_{k-1}$ to $\sum A_k$ never decreases $\sum A_m$.
\end{proposition}

\begin{proof}
Suppose that we make a move using updates for the sum and sum of squares described by \cref{new-sum-relation} and \cref{sum-squares-move} respectively with some value for $z$. Then, we will show that there always exists another value $z'$ with which we can make a move using \cref{sum-move} and \cref{sum-squares-move} that results in the same value of $\sum A_k$ but a larger value of $\sum A_k^2$. By the same logic that was used in the proof of \cref{optimal-bound-sum-squares}, this will never decrease $\sum A_m$.

Since the values of $\sum A_k$ must match, our selection of $z'$ must satisfy
\[\eps z'=\zeta-z=\lfloor z(1+\eps)\rfloor-z.\]
This is clearly positive, so $z'$ must be positive as well. We can also bound the size of $z'$ with
\begin{align*}
    \eps z'&=\lfloor z(1+\eps)\rfloor-z\leq z(1+\eps)-z=\eps z.
\end{align*}
Thus, $z'\leq z$. By \cref{sum-squares-move}, smaller values of $z$ yield larger values of $\sum A_k^2$. This is exactly what we wanted.
\end{proof}

With \cref{optimal-bound-sum-squares} and \cref{optimal-bound-sum}, we now have a strictly stronger set of actions on $\sum A_{k-1}$ and $\sum A_{k-1}^2$ that we can instead use with each of our moves.

\begin{definition}
    A \emph{strong move} with $z$ on $\bp{\sum A_{k-1},\sum A_{k-1}^2}$ makes the following updates:
    \begin{align*}
        \sum A_k&=\sum A_{k-1}+\eps z+1 \\
        \sum A_k^2&=\sum A_{k-1}^2+2\sum A_{k-1}-z^2+2z.
    \end{align*}
\end{definition}

For all integers $1\leq k\leq m$, let $z_k$ denote the value of $z$ used in the strong reversed move from $\bp{\sum A_{k-1},\sum A_{k-1}^2}$ to $\bp{\sum A_k,\sum A_k^2}$. Then,
\begin{align}
    \sum A_k&=\sum A_0+k+\eps\sum_{i=1}^kz_i. \label{sum-alphas}
\end{align}
The same process of expanding is more complicated for $\sum A_m^2$:
\begin{align*}
    \sum A_m^2&=\sum A_0^2+2\sum_{i=1}^m\bp{\sum A_i}-\sum_{i=1}^mz_i^2+2\sum_{i=1}^mz_i.
\end{align*}
We can then rewrite the $\sum A_j$ by plugging in \cref{Em^2-greater-than-Em} and \cref{sum-alphas}, and we thus have
\begin{align}
    \sum A_m^2&=\sum A_0^2+2\sum_{i=1}^m\bp{\sum A_0+i+\eps\sum_{j=1}^iz_j}-\sum_{i=1}^mz_i^2+2\sum_{i=1}^mz_i \nonumber \\
    0&\leq\sum A_0^2+2\sum_{i=1}^m\bp{\eps\sum_{j=1}^mz_j}-\sum_{i=1}^mz_i^2+2\sum_{i=1}^mz_i \nonumber \\
    0&\leq\sum A_0^2+2\eps m\sum_{i=1}^mz_i-\sum_{i=1}^mz_i^2+2\sum_{i=1}^mz_i. \label{big-sum-temp}
\end{align}
By the Cauchy-Schwarz Inequality,
\begin{align*}
    \bp{\sum_{j=1}^m1}\bp{\sum_{j=1}^mz_j^2}&\geq\bp{\sum_{j=1}^mz_j}^2 \\
    \sum_{j=1}^mz_j^2&\geq m^{-1}\bp{\sum_{j=1}^mz_j}^2.
\end{align*}
Plugging this result into \cref{big-sum-temp} yields the following quadratic polynomial in terms of $\sum_{i=1}^mz_i$:
\begin{align*}
    0&\leq\sum A_0^2+(2\eps m+2)\sum_{i=1}^mz_i-m^{-1}\bp{\sum_{j=1}^mz_j}^2 \\
    0&\geq m^{-1}\bp{\sum_{j=1}^mz_j}^2-(2\eps m+2)\sum_{i=1}^mz_i-\sum A_0^2.
\end{align*}
Since the leading coefficient is positive and the quadratic polynomial is on the lesser side of the inequality, $\sum_{j=1}^mz_j$ must be less than the larger root of this polynomial. This larger root is obtained via the quadratic formula:
\begin{align}
    \sum_{j=1}^mz_j\leq\frac{(2\eps m+2)+\sqrt{(2\eps m+2)^2+4m^{-1}\sum A_0^2}}{2m^{-1}}. \label{sum-alpha-bound}
\end{align}

\begin{lemma}\label[lemma]{sum-chippy}
     Suppose that $\eps=\Omega(n^{-1/4})$. Then, $\sum_{i=1}^mz_j$ must be $O\bp{\eps^2n^2|S_0|}$.
\end{lemma}

\begin{proof}
    Since $\eps=\Omega(n^{-1/4})=\Omega(m^{-1/4})$ by \cref{m-equals-n}, the $2\eps m$ in \cref{sum-alpha-bound} must be $\omega(1)$. Thus, if we substitute asymptotic behavior into \cref{sum-alpha-bound}, we get
    \[\sum_{j=1}^mz_j\leq\frac{\Theta(\eps m)+\sqrt{\Theta\bp{\eps^2m^2}+\Theta\bp{m^{-1}\sum A_0^2}}}{\Theta\bp{m^{-1}}}.\]
    By \cref{m-equals-n}, we can replace every $m$ with $n$ to get
    \[\sum_{j=1}^mz_j\leq\frac{\Theta(\eps n)+\sqrt{\Theta\bp{\eps^2n^2}+\Theta\bp{n^{-1}\sum A_0^2}}}{\Theta\bp{n^{-1}}}.\]
    By \cref{E02-less-than-(E0)^2} and \cref{sum-E0}, $\sum A_0^2\leq(\sum A_0)^2=O(n^2|S_0|^2)$, so we can rewrite the previous inequality as follows:
    \begin{align}
        \sum_{j=1}^mz_j\leq\frac{\Theta(\eps n)+\sqrt{\Theta\bp{\eps^2n^2}+O\bp{n|S_0|^2}}}{\Theta\bp{n^{-1}}}. \label{asy-sum-bound}
    \end{align}
    Depending on which one of the two terms inside of the square root dominate, there are two cases that we have to tackle.
    \begin{CaseThm}
        $\Theta\bp{\eps^2n^2}$ dominates $\Theta\bp{n|S_0|^2}$.
    \end{CaseThm}
    In this case, we can rewrite \cref{asy-sum-bound} as follows:
    \begin{align*}
        \sum_{j=1}^mz_j&\leq\frac{\Theta(\eps n)+\sqrt{\Theta\bp{\eps^2 n^2}}}{\Theta\bp{n^{-1}}} \\
        \sum_{j=1}^mz_j&\leq\frac{\Theta(\eps n)}{\Theta\bp{n^{-1}}} \\
        \sum_{j=1}^mz_j&=O\bp{\eps n^2}.
    \end{align*}
    Recall that $|S_0|=\Omega\bp{\eps^{-1}}$ because we must begin with that many bad edges. This means that $O\bp{\eps^2n^2|S_0|}$ must include $O\bp{\eps n^2}$, and we are done with this case.
    \begin{CaseThm}
        $\Theta\bp{n|S_0|^2}$ dominates $\Theta\bp{\eps^2n^2}$.
    \end{CaseThm}
    In this case, we can rewrite \cref{asy-sum-bound} as follows:
    \begin{align*}
        \sum_{j=1}^mz_j&\leq\frac{\sqrt{\Theta(\eps^2 n^2)}+\sqrt{\Theta\bp{\eps^2 n^2}+O\bp{n|S_0|^2}}}{\Theta\bp{n^{-1}}} \\
        \sum_{j=1}^mz_j&\leq\frac{\sqrt{\Theta(\eps^2 n^2)}+\sqrt{O\bp{n|S_0|^2}}}{\Theta\bp{n^{-1}}} \\
        \sum_{j=1}^mz_j&\leq\frac{\sqrt{O\bp{n|S_0|^2}}}{\Theta\bp{n^{-1}}} \\
        \sum_{j=1}^mz_j&=O\bp{n^{3/2}|S_0|}.
    \end{align*}
    Since $\eps=\Omega(n^{-1/4})$, 
    \begin{align*}
        \sum_{j=1}^mz_j&=O\bp{\eps^2\cdot\eps^{-2}n^{3/2}|S_0|} \\
        \sum_{j=1}^mz_j&=O\bp{\eps^2\bp{n^{-1/4}}^{-2}n^{3/2}|S_0|} \\
        \sum_{j=1}^mz_j&=O\bp{\eps^2n^2|S_0|}.
    \end{align*}
\end{proof}

The result of \cref{sum-chippy} can be plugged into \cref{sum-alphas} with \cref{m-equals-n} and \cref{sum-E0} to obtain a bound on $\sum A_m$. Since $\eps=\Omega(n^{-1/4})$, we have
\begin{align}
    \sum A_m&=\sum A_0+m+\eps\sum_{i=1}^mz_i \nonumber \\
    \sum A_m&=O\bp{n|S_0|}+\Theta(n)+O\bp{\eps^3n^2|S_0|} \nonumber \\
    \sum A_m&=O\bp{\eps^3n^2|S_0|}. \label{sum-Em-bound}
\end{align}
By \cref{sum-Em-bound} and \cref{sum-to-finish}, the final number of bad edges in the graph must be $O\bp{\eps^2n^2|S_0|}$. By construction, every vertex in $S_0$ begins with at least one bad edge, so we must begin with $\Omega\bp{|S_0|}$ bad edges. We obtain the following upper bound on the price of uncertainty for consensus games when $\eps=\Omega(n^{-1/4})$:
\[\PoU(\eps,\text{consensus})=\frac{O\bp{\eps^2n^2|S_0|}}{\Omega\bp{|S_0|}}=O\bp{\eps^2n^2}.\]
This finishes the proof of \cref{thm:upper-bound}.

\section{Conclusion}
In this paper, we established tight asymptotic bounds on the worst-case increase in social cost experienced by a network under uncertainty. Our work leaves open some interesting future research directions. Firstly, we believe that the lower bound also holds in the more general case of arbitrary agent update ordering. This may more closely reflect the dynamics of real social networks, and we have a modified version of our construction in this paper that could potentially work for this new problem. Additionally, the study of adversarial influences to normal operations due to uncertainty or individual biases could also be extended to other game-theoretic network settings. While agents in our consensus game model have fixed relationships defined by graph edges, our model does not totally capture the fluidity of social networks due to the formation and destruction of individual connections. Incorporating uncertainty into more general models for social networks can provide greater insight into their behavior in the real world.

\section*{Acknowledgments}
This paper was supported by a NSF Graduate Research Fellowship.

\newpage

\bibliographystyle{plainnat}
\bibliography{references}

\end{document}